\documentclass[11pt]{amsart}
\usepackage{amsmath, amssymb}

\theoremstyle{plain}
\newtheorem{Theorem}{Theorem}

\newtheorem{Proposition}{Proposition}

\theoremstyle{definition}
\newtheorem{Remark}{Remark}

 \DeclareMathOperator{\adj}{adj}
 \DeclareMathOperator{\diag}{diag}
 \DeclareMathOperator{\so}{so}

\title[Elliptic curves and new integrable systems]
{Elliptic curves  and a new construction of integrable systems}
\author{Vladimir Dragovi\'c and Borislav Gaji\'c}
\address{
}

\begin{document}
\abstract A class of elliptic curves with associated Lax matrices
is considered. A family of dynamical systems on $e(3)$
parametrized by polynomial $a$ with above Lax matrices are
constructed. Five cases from the family are selected by the
condition of preserving the standard measure. Three of them are
Hamiltonian. It is proved that two  other cases are not
Hamiltonian in the standard Poisson structure on $e(3)$.
Integrability of all five cases is proven. Integration procedures
are performed in all five cases.  Separation of variables in
Sklyanin sense is also given. A connection with Hess-Appel'rot
system is established. A sort of separation of variables is
suggested for the Hess-Appel'rot system.

\endabstract
\maketitle
\tableofcontents
\section{Introduction}

Starting from a class of elliptic curves we construct a  class of
systems on the Lie algebra $e(3)$ by the Lax representation.
Equations of motion depend on an arbitrary polynomial $a$
and have the form:
$$
\dot x=\{x, H_1\}+a\{x, H_2\}
$$
From the condition that the standard measure is preserved, we
obtain five choices for $a$. In three cases, the systems are Hamiltonian. However, in two other cases they are not Hamiltonian
in the standard Poisson structure on $e(3)$. Integrability in all five cases is performed. For all five cases classical
integration procedure is presented. For the first three cases we gave an
algebro-geometric integration procedure also. Using the Lax matrix,
a separation of variables in Sklyanin sense in the first three cases  is shown. At the and, a connection with the Hess-Appel'rot system is
established. A sort of separation of variables is suggested.

\section{Lax representations and elliptic curves}

Let us start from a class of elliptic curves
\begin{equation}
\mu^2=-p-\frac{a}{\lambda}-\frac{b}{\lambda^2}-\frac{c}{\lambda^3}-\frac{d}{\lambda^4}.
\label{0.1}
\end{equation}
They have the following representation:
\begin{equation}
\mu^2=\frac{\omega^2(\lambda)}{\lambda^4}-2\frac{\Delta(\lambda)\Delta^*(\lambda)}{\lambda^4},
\label{0.2}
\end{equation}
where
$$
\omega(\lambda)=-i(q\lambda^2+x_1\lambda+y_1),\ \ \Delta(\lambda)=x\lambda+y,\ \
\Delta^*(\lambda)=\bar x\lambda+\bar y.
$$
Such curves appear as spectral curves of Lax matrices of the
form
\begin{equation}
L(\lambda)=\left[\begin{matrix}\frac{\omega(\lambda)}{\lambda^2}&\sqrt{2}i\frac{\Delta(\lambda)}{\lambda^2}\\
 & \\
\sqrt{2}i\frac{\Delta^*(\lambda)}{\lambda^2}&-\frac{\omega(\lambda)}{\lambda^2}\end{matrix}\right].
\label{2}
\end{equation}

We will consider a  case that corresponds to real coefficients $a,
b, c, d$.

Let the functions $M_1, M_2, M_3, \Gamma_1, \Gamma_2, \Gamma_3$ be
generators of Lie algebra  $e(3)$  with Poisson structure defined
by the relations
\begin{equation}
\{M_i, M_j\}=-\epsilon_{ijk} M_k,\ \ \{M_i,
\Gamma_j\}=-\epsilon_{ijk} \Gamma_k,\ i,j,k=1,2,3. \label{3}
\end{equation}
Assume the following change of variables:
$$
\begin{aligned}
y&=\frac1{\sqrt{2}}(\beta\Gamma_1-\alpha\Gamma_3-i\Gamma_2),\ \ x=\frac1{\sqrt{2}}(\beta M_1-\alpha M_3-i M_2),\\
\bar{y}&=\frac1{\sqrt{2}}(\beta\Gamma_1-\alpha\Gamma_3+i\Gamma_2), \ \ \bar{x}=\frac1{\sqrt{2}}(\beta M_1-\alpha M_3+i M_2),\\
y_1&=\alpha\Gamma_1+\beta\Gamma_3, \ \ x_1=\alpha M_1+\beta M_3,\ \ q=I_2\sqrt{x_0^2+z_0^2},
\end{aligned}
$$
where $\alpha=\frac{x_0}{\sqrt{x_0^2+z_0^2}},\beta=\frac{z_0}{\sqrt{x_0^2+z_0^2}},$ and $x_0, z_0, I_2$ are constants.

In terms of $x, y, \bar x, \bar y, x_1, y_1$ the Poisson structure
\eqref{3} has the form
$$
\begin{aligned}
\{x,y\}&=0, & \{\bar{x}, \bar{y}\}&=0, & \{x,x_1\}&=ix, & \{\bar{x},x_1\}&=-i\bar{x},& \{y,x_1\}&=iy,\\
\{\bar{y},x_1\}&=-i\bar{y}, & \{y_1, x_1\}&=0, & \{\bar{x},y_1\}&=-i\bar{y},& \{x,y_1\}&=iy, & \{y_1, y\}&=0, \\
 \{y_1, \bar{y}\}&=0, & \{x, \bar{x}\}&=-ix_1, & \{y,\bar{y}\}&=0, & \{x,\bar{y}\}&=-iy_1, & \{\bar{x}, y\}&=iy_1.
\end{aligned}
$$

To each elliptic curve \eqref{0.2} we correspond a family of
dynamical systems:
\begin{equation}
\dot{L}(\lambda)=\frac{1}{2I_2}\left[L(\lambda),\frac{\lambda^2
L(\lambda)-a^2L(a)}{\lambda-a}\right]. \label{4}
\end{equation}
The matrix $L(\lambda)$ in \eqref{4} is of the form \eqref{2} and
$a$ is an arbitrary polynomial in generators of algebra $e(3)$.

Observe that matrices $L$ given by \eqref{2} satisfy
\begin{equation}
\left\{\overset {1}{L}(\lambda), \overset{2}{L}(\mu)\right\}=
\left[r(\lambda-\mu),
\overset{1}{L}(\lambda)+\overset{2}{L}(\mu)\right], \label{4.5}
\end{equation}
where
$$
\overset{1}{L}(\lambda)=L(\lambda)\otimes\left[\begin{matrix} 1&0\\0&1\end{matrix}\right],\
\overset{2}{L}(\mu)=\left[\begin{matrix} 1&0\\0&1\end{matrix}\right]\otimes L(\mu),
$$
with permutation matrix as an $r$-matrix
$$
r(\lambda)=\frac{-1}{\lambda}
\left[\begin{matrix} 1&0&0&0\\
                     0&0&1&0\\
                     0&1&0&0\\
                     0&0&0&1
\end{matrix}\right].
$$
\section{Equations of motion}
The aim of this section is to provide the initial analysis of
dynamical systems \eqref{4}.

The Poisson bracket \eqref{3} is degenerative. As it is well
known, there are two Casimir functions:
\begin{equation}
F_1=M_1\Gamma_1+M_2\Gamma_2+M_3\Gamma_3,\ \ F_2=\Gamma_1^2+\Gamma_2^2+\Gamma_3^2.
\label{4a}
\end{equation}
Thus, a symplectic leaf, defined by conditions  $F_1=c_1, F_2=c_2$
is a four-dimensional manifold. For integrability in Liouville
sense on $e(3)$, one first integral more beside the Hamiltonian is
necessary. On the other hand, if a system is not Hamiltonian,
generally speaking, five first integrals of motion for
integrability in quadratures are required. But, if a
nonhamiltonian system has an  invariant measure, then, according
to the Jacobi theorem, for integrability one needs only four first
integrals of motion.

Equations \eqref{4} can be rewritten in the form
\begin{equation}
\aligned
\dot{M}_1&=z_0\Gamma_2 +az_0M_2,\\
\dot{M}_2&=x_0\Gamma_3-z_0\Gamma_1+a(x_0M_3-z_0M_1),\\
\dot{M}_3&=-x_0\Gamma_2 -ax_0M_2,\\
\dot{\Gamma}_1&=\frac{\Gamma_2M_3-\Gamma_3M_2}{I_2}+az_0\Gamma_2,\\
\dot{\Gamma}_2&=\frac{\Gamma_3M_1-\Gamma_1M_3}{I_2}+a(x_0\Gamma_3-z_0\Gamma_1),\\
\dot{\Gamma}_3&=\frac{\Gamma_1M_2-\Gamma_2M_1}{I_2}-ax_0\Gamma_2.
\endaligned
\label{5}
\end{equation}
We have the following Proposition.
\begin{Proposition}{\label{p1}} System \eqref{5} can be rewritten as:
\begin{equation}
\aligned
\dot{M}_i&=\{M_i, H_1\}+a\{M_i, H_2\},\\
\dot{\Gamma}_i&=\{\Gamma_i, H_1\}+a\{\Gamma_i, H_2\},\quad i=1,2,3,
\endaligned
\label{6}
\end{equation}
where
$$
H_1=\frac{M_1^2+M_2^2+M_3^2}{2I_2}+(x_0\Gamma_1+z_0\Gamma_3),\quad
H_2=x_0M_1+z_0M_3.
$$
\end{Proposition}
For a general polynomial $a$, the system \eqref{5} is neither
Hamiltonian in the Poisson structure \eqref{3}, nor preserves the
standard measure. Simple criterium for preserving the standard
measure is given by:
\begin{Proposition}{\label{p2}} The system \eqref{5} preserves the standard measure if and only if the polynomial
$a$ satisfies the condition:
$$
\{a, x_0M_1+z_0M_3\}=0.
$$
\end{Proposition}
\begin{proof}
The divergence of the vector field in \eqref{5} is equal to $\{a,
x_0M_1+z_0M_3\}$.
\end{proof}

As a consequence of Proposition \ref{p2} we have:

\begin{Proposition}{\label{p2a}}
In the following five cases, the  standard measure is preserved
\begin{description}
\item[(i)] if the polynomial $a$ is a Casimir function:
$a=M_1\Gamma_1+M_2\Gamma_2+M_3\Gamma_3$; \item[(ii)] if the
polynomial $a$ is a Casimir function:
$a=\Gamma_1^2+\Gamma_2^2+\Gamma_3^2$; \item[(iii)] if
$a=x_0M_1+z_0M_3$; \item[(iv)] if $a=x_0\Gamma_1+z_0\Gamma_3$;
\item[(v)]if $a=M_1^2+M_2^2+M_3^2$.
\end{description}
\end{Proposition}
\begin{Theorem}{\label{t1}} If $x_0\ne0$, or $z_0\ne 0$, in the first three cases given above,
the systems are Hamiltonian, while in the fourth and the fifth
case, the systems are not Hamiltonian in the Poisson structure
\eqref{3}.
\end{Theorem}
\begin{proof}
If $a$ is a Casimir function, for an arbitrary function $f$ we
have
$$
\{f,H_1\}+a\{f,H_2\}=\{f,H_1+aH_2\}.
$$
Hence, in the first two cases the systems are Hamiltonian with
Hamiltonian functions
$$
H=H_1+aH_2.
$$
In the third case, since $a=H_2$, we have
$$
\{x^i,H_1\}+H_2\{x^i,H_2\}=\{x^i, H_1+\frac{H_2^2}{2}\},
$$
where $x^i, \ i=1,...,6$ are coordinates $M_1,M_2,M_3,\Gamma_1,
\Gamma_2, \Gamma_3$. Thus, the system is also Hamiltonian with the
Hamiltonian function
$$
H=H_1+\frac{H_2^2}{2}.
$$
The fourth and the fifth cases are  more complicated. For a system
\eqref{5} to be Hamiltonian, the polynomial $a$ needs to satisfy
$$
a\{x^i,H_2\}=\{x^i, H\}
$$
for some function $H$. We get the condition
\begin{equation}
\sum_{j=1}^6\{x^i, x^j\}\left(a\frac{\partial H_2}{\partial x^j}-\frac{\partial H}{\partial x^j}\right)=0.
\label{7}
\end{equation}
In the fourth case, when $a=x_0\Gamma_1+z_0\Gamma_3$, the system of linear partial differential equation
\eqref{7} becomes
\begin{equation}
\begin{aligned}
M_3\frac{\partial H}{\partial M_2}&-M_2\frac{\partial H}{\partial M_3}+x_0z_0\Gamma_1M_2+z_0^2\Gamma_3M_2+
\Gamma_3\frac{\partial H}{\partial\Gamma_2}-\Gamma_2\frac{\partial H}{\partial\Gamma_3}=0,\\
M_1\frac{\partial H}{\partial M_3}&-M_3\frac{\partial H}{\partial M_1}+x_0^2\Gamma_1M_3+x_0z_0\Gamma_3M_3\\
&-x_0z_0\Gamma_1M_1-z_0^2\Gamma_3M_1-\Gamma_3\frac{\partial H}{\partial\Gamma_1}+
\Gamma_1\frac{\partial H}{\partial\Gamma_3}=0,\\
M_2\frac{\partial H}{\partial M_1}&-M_1\frac{\partial H}{\partial M_2}-x_0^2\Gamma_1M_2-x_0z_0\Gamma_3M_2+
\Gamma_2\frac{\partial H}{\partial\Gamma_1}-\Gamma_1\frac{\partial H}{\partial\Gamma_2}=0,\\
\Gamma_3\frac{\partial H}{\partial M_2}&-\Gamma_2\frac{\partial H}{\partial M_3}+x_0z_0\Gamma_1\Gamma_2+
z_0^2\Gamma_3\Gamma_2=0,\\
\Gamma_1\frac{\partial H}{\partial M_3}&-\Gamma_3\frac{\partial H}{\partial M_1}+
x_0^2\Gamma_1\Gamma_3+x_0z_0\Gamma_3^2
-x_0z_0\Gamma_1^2-z_0^2\Gamma_3\Gamma_1=0,\\
\Gamma_2\frac{\partial H}{\partial M_1}&-\Gamma_1\frac{\partial H}{\partial M_2}-
x_0^2\Gamma_1\Gamma_2-x_0z_0\Gamma_3\Gamma_2=0.
\end{aligned}
\label{8}
\end{equation}

It is easy to see that the system \eqref{8} consists of only four
independent equations
\begin{equation}
\begin{aligned}
\Gamma_2\frac{\partial H}{\partial M_1}&=\Gamma_1\frac{\partial H}{\partial M_2}+x_0^2\Gamma_1\Gamma_2
+x_0z_0\Gamma_2\Gamma_3,\\
\Gamma_2\frac{\partial H}{\partial M_3}&=\Gamma_3\frac{\partial H}{\partial M_2}+x_0z_0\Gamma_1\Gamma_2
+z_0^2\Gamma_2\Gamma_3,\\
\Gamma_2\frac{\partial H}{\partial \Gamma_1}&=\Gamma_1\frac{\partial H}{\partial \Gamma_2}+
M_1\frac{\partial H}{\partial M_2}-\frac{\Gamma_1}{\Gamma_2}M_2\frac{\partial H}{\partial M_2},\\
\Gamma_2\frac{\partial H}{\partial \Gamma_3}&=\Gamma_3\frac{\partial H}{\partial \Gamma_2}+
M_3\frac{\partial H}{\partial M_2}-\frac{\Gamma_3}{\Gamma_2}M_2\frac{\partial H}{\partial M_2}.
\end{aligned}
\label{9}
\end{equation}
From the first and the third equation \eqref{9}, one gets
$$
\frac{\partial^2H}{\partial\Gamma_1\partial M_1}-\frac{\partial^2H}{\partial M_1\partial \Gamma_1}=x_0^2.
$$
If $x_0\ne0$, the condition \eqref{7} is not satisfied.
Consequently, the system is not Hamiltonian. In the case when
$z_0\ne 0$ one gets contradiction by observing that
$\frac{\partial^2H}{\partial\Gamma_3\partial
M_3}-\frac{\partial^2H}{\partial M_3\partial \Gamma_3}\ne0$.

Following the same procedure given below, in the fifth case, when $a=M_1^2+M_2^2+M_3^2$ system \eqref{7} becomes
\begin{equation}
\begin{aligned}
M_3\frac{\partial H}{\partial M_2}&-M_2\frac{\partial H}{\partial M_3}+z_0M_2(M_1^2+M_2^2+M_3^2)+
\Gamma_3\frac{\partial H}{\partial\Gamma_2}-\Gamma_2\frac{\partial H}{\partial\Gamma_3}=0,\\
M_1\frac{\partial H}{\partial M_3}&-M_3\frac{\partial H}{\partial M_1}+x_0M_3(M_1^2+M_2^2+M_3^2)\\
&-z_0M_1(M_1^2+M_2^2+M_3^2)-\Gamma_3\frac{\partial H}{\partial\Gamma_1}+
\Gamma_1\frac{\partial H}{\partial\Gamma_3}=0,\\
M_2\frac{\partial H}{\partial M_1}&-M_1\frac{\partial H}{\partial M_2}-x_0M_2(M_1^2+M_2^2+M_3^2)
+\Gamma_2\frac{\partial H}{\partial\Gamma_1}-\Gamma_1\frac{\partial H}{\partial\Gamma_2}=0,\\
\Gamma_3\frac{\partial H}{\partial M_2}&-\Gamma_2\frac{\partial H}{\partial M_3}+z_0\Gamma_2(M_1^2+M_2^2+M_3^2)=0,\\
\Gamma_1\frac{\partial H}{\partial M_3}&-\Gamma_3\frac{\partial H}{\partial M_1}+
(x_0\Gamma_3-z_0\Gamma_1)(M_1^2+M_2^2+M_3^2)=0,\\
\Gamma_2\frac{\partial H}{\partial M_1}&-\Gamma_1\frac{\partial H}{\partial M_2}-
x_0\Gamma_2(M_1^2+M_2^2+M_3^2)=0.
\end{aligned}
\label{10}
\end{equation}
System \eqref{10} has four independent equations
\begin{equation}
\begin{aligned}
\Gamma_2\frac{\partial H}{\partial M_1}&=\Gamma_1\frac{\partial H}{\partial M_2}+x_0\Gamma_2(M_1^2+M_2^2+M_3^2),\\
\Gamma_2\frac{\partial H}{\partial M_3}&=\Gamma_3\frac{\partial H}{\partial M_2}+z_0\Gamma_2(M_1^2+M_2^2+M_3^2),\\
\Gamma_2\frac{\partial H}{\partial \Gamma_1}&=\Gamma_1\frac{\partial H}{\partial \Gamma_2}+
M_1\frac{\partial H}{\partial M_2}-\frac{\Gamma_1}{\Gamma_2}M_2\frac{\partial H}{\partial M_2},\\
\Gamma_2\frac{\partial H}{\partial \Gamma_3}&=\Gamma_3\frac{\partial H}{\partial \Gamma_2}+
M_3\frac{\partial H}{\partial M_2}-\frac{\Gamma_3}{\Gamma_2}M_2\frac{\partial H}{\partial M_2}.
\end{aligned}
\label{11}
\end{equation}
From the first and the third equation \eqref{11}, one get
$$
\frac{\partial^2H}{\partial\Gamma_1\partial M_1}-\frac{\partial^2H}{\partial M_1\partial \Gamma_1}=
2x_0\frac{M_2}{\Gamma_2}(M_1\Gamma_2-M_2\Gamma_1).
$$
Consequently, in the fifth case system is not Hamiltonian.
\end{proof}
Regarding integrability of given five cases, we have simple Proposition.
\begin{Proposition}{\label{p3}}
\item{\rm{(a)}} A function $F$ is a first integral of equations \eqref{5} if it satisfies
$$
\dot{F}=\{F,H_1\}+a\{F,H_2\}=0.
$$
\item{\rm{(b)}} The Casimir functions $F_1$ and $F_2$ and
functions $H_1$ and $H_2$ are integrals of system \eqref{5} for
any polynomial $a$.
\end{Proposition}
Finally we have:
\begin{Theorem}{\label{t2}} The system \eqref{5} in cases (i)-(iii) is completely integrable in Liouville sense.
In cases (iv) and (v), system \eqref{5} is integrable in quadratures.
\end{Theorem}

\section{Algebro-geometric integration procedure of the systems}

We pass to an algebro-geometric integration procedure, based on
construction of the Baker-Akhiezer vector-function. For more
details about the Baker-Akhiezer functions see \cite{D1, D2, DKN}.

The dynamical systems have Lax representation
$$
\dot{L}(\lambda)=[L(\lambda), A(\lambda)],
$$
with $L(\lambda)$ given by \eqref{2} and
$$
A(\lambda)=\frac{\lambda^2 L(\lambda)-a^2L(a)}{2I_2(\lambda-a)}.
$$
The corresponding spectral curve $\Gamma$ is elliptic curve \eqref{0.2}.

As  usual, we consider the following eigenvalue problem
\begin{equation}
\begin{aligned}
&\left(\frac{d}{dt}+A(\lambda)\right)\Psi(t,P)=0,\\
&L(\lambda)\Psi(t,P)=\mu\Psi(t,P),
\end{aligned}
\label{4.1}
\end{equation}
with a normalization
\begin{equation}
\Psi^1(0,P)+\Psi^2(0,P)=1,
\label{4.2}
\end{equation}
where $P=(\lambda, \mu)$ is a point on the spectral curve
$\Gamma$.

Let us denote by $\Phi(t,\lambda)$ the fundamental solution of the
system
$$
\left(\frac{d}{dt}+A(\lambda)\right)\Phi(t,\lambda)=0,
$$
normalized by the condition $\Phi(0,\lambda)=1$.

Let us denote by $\infty^+$ and $\infty^-$ the two points on the
curve $\Gamma$ over $\lambda=\infty$, with
$\mu=iI_2\sqrt{x_0^2+z_0^2}$ and $\mu=-iI_2\sqrt{x_0^2+z_0^2}$
respectively.
\begin{Proposition}{\label{p4.1}} If the polynomial $a$ is  a first integral of  motion, then
the vector-function $\Psi(t,P)$ satisfies the following
conditions:
\begin{itemize}
\item[(a)] In the affine part of the curve $\Gamma$, the
vector-function $\Psi(t,P)$ has two time independent poles, and
each of the components $\Psi^1(t,P)$ and $\Psi^2(t,P)$ has one
zero. \item[(b)] At the points $\infty^+$ and $\infty^-$, the
functions $\Psi^1$ and $\Psi^2$ have essential singularities with
the following asymptotics:
$$
\begin{aligned}
\Psi^1(t,P)&=\begin{cases}e^{\frac i2\left(\sqrt{x_0^2+z_0^2}\ (\lambda+a)+\frac{x_1}{I_2}\right)t}
\left(1+O(\frac 1\lambda)\right), \quad &P\rightarrow \infty^-\\
e^{-\frac i2\left(\sqrt{x_0^2+z_0^2}\ (\lambda+a)+\frac{x_1}{I_2}\right)t}
\left(O(\frac 1\lambda)\right), \quad &P\rightarrow \infty^+
\end{cases}\\
\Psi^2(t,P)&=\begin{cases}e^{\frac i2\left(\sqrt{x_0^2+z_0^2}\ (\lambda+a)+\frac{x_1}{I_2}\right)t}
\left(O(\frac 1\lambda)\right), \quad &P\rightarrow \infty^-\\
e^{-\frac i2\left(\sqrt{x_0^2+z_0^2}\ (\lambda+a)+\frac{x_1}{I_2}\right)t}
\left(1+O(\frac 1\lambda)\right), \quad &P\rightarrow \infty^+
\end{cases}
\end{aligned}
$$
\item[(c)] The asymptotics have the form
$$
\begin{aligned}
\Psi^1(t,P)&=e^{-\frac i2\left(\sqrt{x_0^2+z_0^2}\ (\lambda+a)+\frac{x_1}{I_2}\right)t}
\left(\frac{x}{I_2\sqrt{2}\sqrt{x_0^2+z_0^2}}\frac 1\lambda+O(1/\lambda^2)\right), \quad &P\rightarrow \infty^+\\
\Psi^2(t,P)&=e^{\frac i2\left(\sqrt{x_0^2+z_0^2}\ (\lambda+a)+\frac{x_1}{I_2}\right)t}
\left(-\frac{\bar{x}}{I_2\sqrt{2}\sqrt{x_0^2+z_0^2}}\frac 1\lambda+O(1/\lambda^2)\right), \quad &P\rightarrow \infty^-\\
\end{aligned}
$$
\end{itemize}
\end{Proposition}
\begin{proof} Since
$$
\Psi(t,P)=\Phi(t,\lambda)\Psi(0,P), \ \ \Phi(0,\lambda)=1,
$$
the poles of $\Psi(t,P)$ coincides with poles of $\Psi(0,P)$
($\Phi(t,\lambda)$ is holomorphic), so they are time-independent.
Using normalization \eqref{4.2}, and \eqref{4.1} one can calculate
that $\Psi(0,P)$ has two poles $P_1$ and $P_2$ in the affine part
of $\Gamma$. One can also conclude, using $d\ln\Psi_1$ and $d\ln\Psi_2$ that
each of the components $\Psi_1$ and $\Psi_2$ has a one zero. This proves part (a).

In order to find asymptotic of $\Psi^1$ at $\infty^+$ and
$\infty^-$ one needs to consider
$$
{(\ln(\Psi^1))}^\cdot=\frac{\dot\Psi^1}{\Psi^1}=-A_{11}-A_{12}\frac{\Psi^2}{\Psi_1}.
$$

Since
$\frac{\Psi^2}{\Psi^1}=\frac{\mu-L_{11}}{L_{12}}=\frac{L_{21}}{\mu-L_{22}}$,
and using asymptotics
$$
\mu=i\left[I_2\sqrt{x_0^2+z_0^2}+\frac{x_1}{\lambda}+O(1/\lambda)\right],\ \ P\rightarrow \infty^+,
$$
and
$$
\mu=-i\left[I_2\sqrt{x_0^2+z_0^2}+\frac{x_1}{\lambda}+O(1/\lambda)\right],\ \ P\rightarrow \infty^-,
$$
we have:
$$
{(\ln(\Psi^1))}^\cdot=-\frac{i}{2}\sqrt{x_0^2+z_0^2}\ (\lambda-a)-
\frac{ix_1}{2I_2}+i\sqrt{x_0^2+z_0^2}\ \frac{y}{x}+O(1/\lambda), \ \ P\rightarrow\infty^+,
$$
and
$$
{(\ln(\Psi^1))}^\cdot=\frac{i}{2}\sqrt{x_0^2+z_0^2}\ (\lambda+a)+
\frac{ix_1}{2I_2}+O(1/\lambda), \ \ P\rightarrow\infty^-.
$$
Since
$$
\frac{\dot{x}}{x}=ia\sqrt{x_0^2+z_0^2}+i\sqrt{x_0^2+z_0^2}\ \frac{y}{x},
$$
we  finally get:
$$
{(\ln(\Psi^1))}^\cdot=-\frac{i}{2}\sqrt{x_0^2+z_0^2}\ (\lambda+a)-
\frac{ix_1}{2I_2}+\frac{\dot x}{x}+O(1/\lambda), \ \ P\rightarrow\infty^+,
$$
and
$$
{(\ln(\Psi^1))}^\cdot=\frac{i}{2}\sqrt{x_0^2+z_0^2}\ (\lambda+a)+
\frac{ix_1}{2I_2}+O(1/\lambda), \ \ P\rightarrow\infty^-.
$$
Similarly, for the function $\Psi^2$ we have:
$$
{(\ln(\Psi^2))}^\cdot=-\frac{i}{2}\sqrt{x_0^2+z_0^2}\ (\lambda+a)-
\frac{ix_1}{2I_2}+O(1/\lambda), \ \ P\rightarrow\infty^+,
$$
and
$$
{(\ln(\Psi^2))}^\cdot=\frac{i}{2}\sqrt{x_0^2+z_0^2}\ (\lambda+a)+
\frac{ix_1}{2I_2}+\frac{\dot{\bar{x}}}{\bar{x}}+O(1/\lambda), \ \ P\rightarrow\infty^-.
$$
This proves parts (b) and (c).
\end{proof}

\begin{Remark} The vector-function $\Psi(t,P)$ satisfies standard conditions of the 2-point Baker-Akhiezer function.
From the parts (a) and (b) of Proposition \ref{p4.1} we can
construct $\Psi(t,P)$. Using the part (c), one can reconstruct
$x(t)$ (and $\bar{x}(t))$ from the Baker-Akhiezer function.
\end{Remark}
\begin{Remark}
All formulae from the proof of Proposition \ref{p4.1} are still
valid even if the polynomial $a$ is not an integral of  motion.
But in that case    statements (b) and (c) of the Proposition are
not valid any more.
\end{Remark}

Now we will give explicit formulae for the Baker-Akhiezer function
in terms of the Jacobi theta-function $\theta_{11}(z|\tau)$ with
characteristics $[\frac12,\frac12]$. Explicit formulas are similar
to the Lagrange case (see \cite{GZ}).

Let us fix the canonical basis of cycles $A$ and $B$ on $\Gamma$
($A\cdot B=1$), and let $\omega$ be the holomorphic differential
normalized by the conditions
$$
\oint_A\omega=2i\pi,\ \ \oint_B\omega=\tau.
$$
Theta-function $\theta_{11}(z|\tau)$ is defined by the relation
$$
\theta_{11}(z|\tau)=\sum_{-\infty}^{\infty} \exp\left[ \frac 12\tau(n+\frac12)+(z+i\pi)(n+\frac12)\right]
$$
and satisfies
$$
\theta_{11}(0)=0,\ \ \theta_{11}(z+2\pi i)=-\theta_{11}(z),\ \ \theta_{11}(z+\tau)=-\exp(-z-\frac12\tau)\theta_{11}(z).
$$

Let $\Omega^+$ and $\Omega^-$ be differentials of the second kind
with principal parts $-\frac{i}{2}\sqrt{x_0^2+z_0^2}\ d\lambda$ and
$+\frac{i}{2}\sqrt{x_0^2+z_0^2}\ d\lambda$ at $\infty^+$ and at
$\infty^-$ respectively, normalized by the condition that
$A$-periods are zero. Let us introduce differential $\Omega=\Omega^++\Omega^-$. We will denote by $U$ the
$B$-period of differential $\Omega$, and by
$c^+$ and $c^-$ the constants:
$$
\begin{aligned}
\int_{P_0}^P\Omega&=-\frac{i}{2}\sqrt{x_0^2+z_0^2}\ \lambda+c^++O(1/\lambda),\ \ P\rightarrow \infty^+\\
\int_{P_0}^P\Omega&=+\frac{i}{2}\sqrt{x_0^2+z_0^2}\ \lambda+c^-+O(1/\lambda),\ \ P\rightarrow \infty^-.
\end{aligned}
$$
\begin{Proposition} The Baker-Akhiezer functions are given by
$$
\begin{aligned}
\Psi^1(t,P)&=c_1\exp\left[(\int_{P_0}^P\Omega-c^-+\frac{i}{2}a+\frac{i}{2}\frac{x_1}{I_2})t\right]
\frac{\theta_{11}(\mathcal{A}(P+\infty^+-P_1-P_2)+tU)}{\theta_{11}({\mathcal{A}(\infty^++\infty^--P_1-P_2)+tU})},\\
\Psi^2(t,P)&=c_2\exp\left[(\int_{P_0}^P\Omega-c^+-\frac{i}{2}a-\frac{i}{2}\frac{x_1}{I_2})t\right]
\frac{\theta_{11}(\mathcal{A}(P+\infty^--P_1-P_2)+tU)}{\theta_{11}({\mathcal{A}(\infty^++\infty^--P_1-P_2)+tU})},\\
\end{aligned}
$$
where constants $c_1$ and $c_2$ are
$$
\begin{aligned}
c_1&=\frac{\theta_{11}(\mathcal{A}(P-\infty^+))\theta_{11}(\mathcal{A}(\infty^--P_1))
\theta_{11}(\mathcal{A}(\infty^--P_2))}{\theta_{11}(\mathcal{A}(\infty^--\infty^+))
\theta_{11}(\mathcal{A}(P-P_1))\theta_{11}(\mathcal{A}(P-P_2))},\\
c_2&=\frac{\theta_{11}(\mathcal{A}(P-\infty^-))\theta_{11}(\mathcal{A}(\infty^+-P_1))
\theta_{11}(\mathcal{A}(\infty^+-P_2))}{\theta_{11}(\mathcal{A}(\infty^+-\infty^-))
\theta_{11}(\mathcal{A}(P-P_1))\theta_{11}(\mathcal{A}(P-P_2))},\\
\end{aligned}
$$
and $\mathcal{A}$ is the Abel map, and $P_1$ and $P_2$ are the
poles of the  function $\Psi$.
\end{Proposition}
\begin{proof} The proof is based on general theory (see \cite{D1,D2,DKN,GZ}).
\end{proof}

\section{Classical integration of the systems}

We are going to show that in all considered five cases integration
of the systems can be reduced to  elliptic integrals.

Let us  change coordinates by the rotation of $xOz$ plane:
$$
\begin{alignedat}{3}
X_1&=\alpha M_1+\beta M_3 & \quad X_2&=M_2 & \quad X_3&=-\beta M_1+\alpha M_3\\
Y_1&=\alpha \Gamma_1+\beta \Gamma_3 & Y_2&=\Gamma_2 & Y_3&=-\beta \Gamma_1+\alpha \Gamma_3
\end{alignedat}
$$

Differential equations \eqref{5} of motion become
\begin{equation}
\begin{aligned}
\dot{X}_1&=0\\
\dot{X}_2&=\sqrt{x_0^2+z_0^2}\,(Y_3+aX_3)\\
\dot{X}_3&=-\sqrt{x_0^2+z_0^2}\,(Y_2+aX_2)\\
\dot{Y}_1&=\frac{1}{I_2}(X_3Y_2-X_2Y_3)\\
\dot{Y}_2&=\frac{1}{I_2}(X_1Y_3-X_3Y_1)+a\sqrt{x_0^2+z_0^2}\,Y_3\\
\dot{Y}_3&=\frac{1}{I_2}(X_2Y_1-X_1Y_2)-a\sqrt{x_0^2+z_0^2}\,Y_2.
\end{aligned}
\label{16}
\end{equation}

The first integrals are
\begin{equation}
\begin{aligned}
F_1&=X_1Y_1+X_2Y_2+X_3Y_3=c_1\\
F_2&=Y_1^2+Y_2^2+Y_3^2=c_2\\
H_1&=\frac{X_1^2+X_2^2+X_3^2}{2I_2\sqrt{x_0^2+z_0^2}}+Y_1=d_1\\
H_2&=X_1=d_2.
\end{aligned}
\label{17}
\end{equation}
From \eqref{16} we have:
$$
\begin{aligned}
\dot{X}_2^2+\dot{X}_3^2&=(x_0^2+z_0^2)\Big[(Y_2^2+Y_3^2)+a^2(X_2^2+X_3^2)+2a(X_2Y_2+X_3Y_3)\Big]\\
\dot{X}_2X_{3}-\dot{X}_3X_2&=\sqrt{x_0^2+z_0^2}\ \Big[X_2Y_2+X_3Y_3+a(X_3^2+X_2^2)\Big].
\end{aligned}
$$
Using the first integrals \eqref{17}, one gets
$$
\begin{aligned}
\dot{X}_2^2+\dot{X}_3^2&=(x_0^2+z_0^2)\Big[c_2-Y_1^2+a^2(X_2^2+X_3^2)+2a(c_1-d_2Y_1)\Big]\\
\dot{X}_2X_{3}-\dot{X}_3X_2&=\sqrt{x_0^2+z_0^2}\ \Big[c_1-d_2Y_1+a(X_3^2+X_2^2)\Big].
\end{aligned}
$$
Introducing  polar coordinates
$$
X_2=\rho\cos{\sigma},\quad X_3=\rho\sin{\sigma}
$$
the last equations become
\begin{equation}
\begin{aligned}
\dot{\rho}^2+\rho^2\dot{\sigma}^2&=(x_0^2+z_0^2)\left[c_2-Y_1^2+a^2\rho^2+2a(c_1-d_2Y_1)\right]\\
-\rho^2\dot{\sigma}&=\sqrt{x_0^2+z_0^2}\left[c_1-d_2Y_1+a\rho^2\right].
\end{aligned}
\label{18}
\end{equation}
 Let us subtract the
square of the second equation multiplied by 4 from the first equation multiplied
by $4\rho^2$. After simplifying one gets
$$
\left[\frac{d}{dt}(\rho^2)\right]^2=4(x_0^2+z_0^2)\left[\rho^2(c_2-Y_1^2)-(c_1-d_2Y_1)^2\right].
$$
Finally, using
\begin{equation}
Y_1=d_1-\frac{d_2^2+\rho^2}{2I_2\sqrt{x_0^2+z_0^2}}
\label{18a}
\end{equation}
and denoting $\rho^2=u$, we get
\begin{equation}
\dot{u}^2=-\frac{u^3}{I_2^2}-Bu^2-Cu-D
\label{19}
\end{equation}
where
$$
\begin{aligned}
B&=\frac{-4A\sqrt{x_0^2+z_0^2}}{I_2}+\frac{d_2^2}{I_2^2}\\
C&=4(x_0^2+z_0^2)\Big(A^2-c_2+\frac{d_2(c_1-d_2A)}{I_2\sqrt{x_0^2+z_0^2}}\Big)\\
D&=4(x_0^2+z_0^2)(c_1-d_2A)^2\\
A&=d_1-\frac{d_2^2}{2I_2\sqrt{x_0^2+z_0^2}}.
\end{aligned}
$$
So, the following proposition is proved:
\begin{Proposition} The function $u(t)$ is an elliptic function of time.
\end{Proposition}
Let us remark that $u$ (and consequently $\rho$) does not depend
of a choice of the polynomial $a$.

Having $u(t)$ as a known function of time, one can find $\rho(t)$ as a known function of time. In order to
reconstruct $X_2$ and $X_3$, one needs to find $\sigma$ as a function of time.

From the second equation of \eqref{18}, using \eqref{18a}, we have
$$
\dot{\sigma}=-\frac{1}{\rho^2(t)}\sqrt{x_0^2+z_0^2}\left[c_1-d_2\left(d_1-\frac{d_2^2+\rho^2(t)}{2I_2\sqrt{x_0^2+z_0^2}}\right)+a\,\rho^2(t)\right]
$$
The right hand side of the last equation is a function of time and
of the polynomial $a$. When $a$ is a first integral of  motion (in
the first three cases defined above) e.q. when $a=const$, then
right hand side of the last equation is a known function of time
and one can find $\sigma$ by quadratures. In the fourth case
$$
a=x_0\Gamma_1+z_0\Gamma_3=\sqrt{x_0^2+z_0^2}\,Y_1=\sqrt{x_0^2+z_0^2}\,d_1-\frac{d_2^2+\rho^2(t)}{2I_2\sqrt{x_0^2+z_0^2}}.
$$
So, in this case $a(t)$ is a known function of time and one can
find $\sigma$ by solving a differential equation. Similarly, in
the fifth case
$$
a=M_1^2+M_2^2+M_3^2=X_1^2+X_2^2+X_3^2=d_2^2+\rho^2(t)
$$
is again a known function of time and a differential equation for
determining  $\sigma$ can be solved. Knowing $\rho$ and $\sigma$
as functions of time, one can reconstruct $X_2$ and $X_3$. From
\eqref{18a}, one finds $Y_1$ as a function of time. Finally, using
the differential equation for $Y_1$ from \eqref{16}, and the
second first integral $F_2=c_2$ from \eqref{17} one can
reconstruct $Y_2$ and $Y_3$.

Two elliptic curves appeared here. The first one $\Gamma$, has
been defined by the equation \eqref{0.1}, and it was the curve
from witch we started. The other one $\Gamma'$, given by
\begin{equation}
v^2=-\frac{u^3}{I_2^2}-Bu^2-Cu-D
\label{19a}
\end{equation}
corresponds to the solution of the differential equation
\eqref{19}. A natural question is how these two curves are
related. We have the following proposition:

\begin{Proposition} The elliptic curves $\Gamma$, defined by the equation \eqref{0.1} and $\Gamma'$ defined by \eqref{19a} are
isomorphic.
\end{Proposition}
\begin{proof}
A direct calculation gives that $j$-invariant of both curves are
the same, so they are isomorphic.
\end{proof}

\section{Separation of variables}

One of the oldest methods in the theory of integrable dynamical
systems is separation of variables. Originally, this method was
built in order to find exact solutions of the Hamilton-Jacobi
equations. In the middle of 1990's, Sklyanin in his celebrated
paper \cite{S1}, introduced a new concept of separation of
variables which was related to modern techniques in the theory of
integrable systems such as the inverse scattering method and the
Lax representation.

If a completely integrable system on $2n$-dimensional symplectic
manifold is given with $n$ functionally independent commuting
first integrals of  motion  $F_1, ..., F_n$, then variables
$\lambda_1,..., \lambda_n, \mu_1,...,\mu_n$ are {\it separated} if
they are canonical e.q.
$$
\{\lambda_i,\mu_j\}=\delta_{ij},\ \ \{\lambda_i, \lambda_j\}=\{\mu_i, \mu_j\}=0,\ \ i,j=1...n,
$$
and if there exist $n$ relations $\Phi_i$ such that
\begin{equation}
\Phi_i(\lambda_i, \mu_i, F_1,...,F_n)=0,\quad i=1,...,n.
\label{s1}
\end{equation}

The Sklyanin {\it magic recipe} gives  separation variables
$\lambda_i$ as poles of the properly normalized Baker-Akhiezer
function. The canonically conjugated variables are corresponding
eigenvalues of the Lax matrix $L(\lambda_i)$ (for details see
\cite{S1}).

In the case of algebra $e(3)$, since symplectic leaves  are
four-dimensional, one needs to find four separation variables
$\lambda_1, \lambda_2, \mu_1, \mu_2$.

\begin{Proposition}{\label{p5}} For the Hamiltonian systems (i), (ii), (iii) defined in Proposition
\ref{p2a} the separation variables are
\begin{equation}
\begin{aligned}
\lambda_1&=\frac{y}{x},\ \ \mu_1=i\left(I_2\sqrt{x_0^2+z_0^2}-\frac{(\alpha M_1+\beta M_3)x}{y} +
\frac{(\alpha \Gamma_1+\beta \Gamma_3)x^2}{y^2}\right),\\
\lambda_2&=x,\ \ \mu_2=-\frac{i}{x}(\alpha M_1+\beta M_3).
\end{aligned}
\label{s3}
\end{equation}
Corresponding separation relations are:
\begin{equation}
\mu_1^2=\frac{\omega(-\lambda_1)^2-2\Delta(-\lambda_1)\Delta^*(-\lambda_1)}{\lambda_1^4},\
\  \lambda_2\mu_2=-iH_2=const. \label{s4}
\end{equation}
\end{Proposition}

\begin{proof}
Following \cite{S1, KNS}, the separation variables for the systems
\eqref{2} satisfy the equation
\begin{equation}
(1\ 0)\adj(L(\lambda)-\mu\cdot 1)=0,
\label{s2}
\end{equation}
which corresponds to the standard normalization
$\vec{\alpha}_0=(1\ 0)$. From \eqref{s2} one  gets only one pair
of separation variables $\lambda_1, \mu_1$. The other two
variables $\lambda_2, \mu_2$ are calculated from the asymptotics
of the matrix $L(\lambda)$ when $\lambda$ goes to  infinity.

By direct calculations we see that variables are canonical.
\end{proof}

In terms of $H_1, H_2, F_1, F_2$, the separation relations can be
rewritten in the form:
$$
\begin{aligned}
\mu_1^2=\frac{1}{\lambda_1^4}\big(-I_2^2(x_0^2+z_0^2)\lambda_1^4+&2I_2\sqrt{x_0^2+z_0^2}\ H_2\lambda_1^3
-2I_2H_1\lambda_1^2+2F_1\lambda_1-F_2\big),\\
H_2&=i\lambda_2\mu_2.
\end{aligned}
$$

\section{Connection with Hess-Appel'rot system}

In this section we find  a sort of separation of variables for the
Hess-Appel'rot case of motion of heavy rigid body about fixed
point.

The equations of motion are the Euler-Poisson equations:
\begin{equation}
\dot{M}=M\times \Omega+\Gamma\times\chi,\quad \dot{\Gamma}=\Gamma\times\Omega.
\label{12}
\end{equation}
Here $M$ is the kinetic momentum vector, $\Omega$ is the vector of
angular velocity, $\Gamma$ is the unit vertical vector, and
$\chi=(x_0,y_0,z_0)$ is the radius vector of the center of masses.
Connection between $M$ and $\Omega$ is given by $M=I\Omega$, where
$I$ is the inertia operator, and we can assume that
$I=\diag(I_1,I_2,I_3)$.

Equations \eqref{12} are Hamiltonian on the Lie algebra $e(3)$ in
the standard Poisson structure given by \eqref{3} with the
Hamiltonian function:
\begin{equation}
H_{RB}=\frac12\langle M,\Omega\rangle+\langle \Gamma, \Omega\rangle.
\label{13}
\end{equation}
Hence, for integrability in the Liouville sense of equations
\eqref{12}, one needs one first integral of motion more.

The Hess-Appel'rot case of rigid body is introduced by Hess and Appel'rot (see \cite{Ap, H}). It is defined by
conditions
\begin{equation}
\begin{aligned}
y_0&=0,\\
x_0\sqrt{I_1(I_2-I_3)}&+z_0\sqrt{I_3(I_1-I_2)}=0.
\end{aligned}
\label{14}
\end{equation}
Instead of the fourth integral, in the Hess-Appel'rot case there
is an invariant relation
\begin{equation}
x_0M_1+z_0M_3=0.
\label{15}
\end{equation}

In \cite{DG1} (see also \cite{DG2}) a Lax representation for the
Hess-Appel'rot system was found. The spectral curve is reducible.
One component is a rational curve, and the other one is an
elliptic curve. In order to find a sort of separation of variables
for this specific situation, we give first another Lax
representation. It is based on following observation:
\begin{Proposition}{\label{p4}} On hypersurface \eqref{15} the equations of the Hess-Appel'rot system
are equivalent to the Lax representation \eqref{2} with
$a=\frac{\alpha\Omega_1+\beta\Omega_2}{\sqrt{x_0^2+z_0^2}}$.
\end{Proposition}

This Lax representation for the Hess-Appel'rot system was given in
\cite{G}. The Lax representation of the same form for the Lagrange
top was constructed in \cite{KPR}.

As a consequence of Propositions \ref{p4} and \ref{p5}, we have:

\begin{Proposition}{\label{p6}} The separation variables for the  Hess-Appel'rot system are
\begin{equation}
\begin{aligned}
\lambda_1&=\frac{y}{x},\ \
\mu_1=i\left(I_2\sqrt{x_0^2+z_0^2}-\frac{(\alpha M_1+\beta
M_3)x}{y} +
\frac{(\alpha \Gamma_1+\beta \Gamma_3)x^2}{y^2}\right),\\
\lambda_2&=x,\ \ \mu_2=-\frac{i}{x}(\alpha M_1+\beta M_3).
\end{aligned}
\label{s5}
\end{equation}
Corresponding separation relations are:
\begin{equation}
\mu_1^2=\frac{\omega(-\lambda_1)^2-2\Delta(-\lambda_1)\Delta^*(-\lambda_1)}{\lambda_1^4},\
\  \mu_2=0. \label{s6}
\end{equation}
\end{Proposition}

In terms of $H_1, H_2, F_1, F_2$ separation relations can be
rewritten in the form:
$$
\begin{aligned}
\mu_1^2=\frac{1}{\lambda_1^4}\big(-I_2^2(x_0^2+z_0^2)\lambda_1^4+&2I_2\sqrt{x_0^2+z_0^2}\ H_2\lambda_1^3
-2I_2H_1\lambda_1^2+2F_1\lambda_1-F_2\big),\\
i\lambda_2\mu_2=0.
\end{aligned}
$$

\section{Acknowledgments}

The research was partially supported by the Serbian Ministry of Sciences and Technology,
Project {\it Geometry and Topology of Manifolds and Integrable Dynamical Systems}. A part of the
article has been written during the visit of one of the authors (V. D.) to the IHES in Autumn 2008, and he
uses the opportunity to thank the IHES for hospitality and outstanding working conditions.

\

\small

\sc

Vladimir Dragovi\' c

 Mathematical Institute  SANU

 Kneza Mihaila 36, 11000 Belgrade, Serbia

and

Mathematical Physics Group, University of Lisbon, Portugal

{\rm e-mail: vladad@mi.sanu.ac.rs}

\

Borislav Gaji\' c

Mathematical Institute  SANU

Kneza Mihaila 36, 11000 Belgrade, Serbia

{\rm e-mail: gajab@mi.sanu.ac.rs}


\begin{thebibliography}{99}

\bibitem{Ap} G. G. Appel'rot: The problem of motion of a rigid body
about a fixed point, {\it Uchenye Zap. Mosk. Univ. Otdel. Fiz.
Mat. Nauk}, No. 11, (1894), 1-112.

\bibitem{Ar} V. I. Arnol'd: {\it Mathematical methods of classical
mechanics}, (Moscow: Nauka, 1989 [in Russian, 3-rd edition]).

\bibitem{AKN} V. I. Arnol'd , V. V. Kozlov and A. I. Neishtadt: {\it
Mathematical aspects of classical and celestial mechanics/ in
Dynamical systems III}, (Berlin: Springer-Verlag, 1988).

\bibitem{DG1} V. Dragovi\' c, B. Gaji\' c: An L-A pair for the
Hess-Apel'rot system and a new integrable case for the
Euler-Poisson equations on $\so(4) \times \so(4)$, {\it Roy. Soc. of
Edinburgh: Proc A}, {\bf 131} (2001), 845-855.

\bibitem{DG2} V. Dragovi\' c, B. Gaji\' c: The Lagrange bitop on
$\so(4)\times \so(4)$ and geometry of Prym varieties,
{\it American Journal of Mathematics}, {\bf 126}, (2004), 981-1004.

\bibitem{DG3} V. Dragovi\' c, B. Gaji\' c: Systems of Hess-Appel'rot type,
{\it Comm. Math. Phys}, {\bf 265}, (2006), 397-435.

\bibitem{D1} B. A. Dubrovin: Completely integrable Hamiltonian
systems connected with matrix operators and Abelian varieties,
{\it Func. Anal. and its Appl.}, {\bf 11} (1977), 28-41, [in
Russian].

\bibitem{D2} B. A. Dubrovin: Theta-functions and nonlinear
equations, {\it Uspekhi Math. Nauk}, {\bf 36} (1981), 11-80, [in
Russian].

\bibitem{DKN} B. A. Dubrovin, I. M. Krichever and S. P. Novikov:
Integrable systems I. {\it in Dynamical systems IV}, (Berlin:
Springer-Verlag, ) 173-280.

\bibitem{G} B. Gaji\' c: {\it Integration of Euler-Poisson
equations by algebro-geometric methods}, Ph.D. thesis, Faculty of
mathematics, Belgrade, 2002 (in Serbian).

\bibitem{GZ} L. Gavrilov, A. Zhivkov: The complex geometry of
Lagrange top. {\it L'Ensei\-gnement Math\' ematique}, {\bf 44}
(1998), 133-170.

\bibitem{Go} V. V. Golubev. {\it Lectures on integration of the equations of
motion of a rigid body about a fixed point} (Moskow: Gostenhizdat, 1953
 [in Russian]; English translation: Philadelphia: Coronet Books, 1953).

\bibitem{H} W. Hess: Ueber die Euler'schen Bewegungsgleichungen und
{\"u}ber eine neue particul{\"a}re L{\"o}sung des Problems der
Bewegung eines starren K{\"o}rpers um einen festen, {\it Punkt.
Math. Ann.}, {\bf 37} (1890), 178-180.

\bibitem{KNS} V. B. Kuznetsov, F. W. Nijhoff, E. K. Sklyanin: Separation of variables
for te Ruijsenaards system. {\it Comm. Math. Physics}, {\bf 189}, (1997), 855-877.

\bibitem{KPR} V. B. Kuznetsov, M. Petrera, O. Ragnisco: Separation of variables
an Backlund transformations for the symmetric Lagrange top.
{\it J. Phys. A}, {\bf 37}, (2004), 8495-8512.


\bibitem{S1} E. K. Sklyanin: Separation of variables. New trends,
{\it Progr. Theor. Phys. Suppl.}, {\bf 118}, (1995), 35-60.

\end{thebibliography}
\end{document}